\documentclass[11pt]{article}

\usepackage[T1]{fontenc}
\usepackage[utf8]{inputenc}
\usepackage{amsmath,amssymb,amsthm,mathtools,bm}
\usepackage[round,authoryear]{natbib}
\usepackage[hidelinks]{hyperref}
\usepackage[margin=1.1in]{geometry}

\newtheorem{theorem}{Theorem}
\newtheorem{corollary}[theorem]{Corollary}

\newtheorem{proposition}[theorem]{Proposition}

\title{DisclosureBeta: A Measurement-Channel Theory for
       Regime-Conditioned Betas from LLM-Read Risk Disclosures}
\author{Wong Ping Kuen (Ricky)}
\date{July 2026 --- theory preprint (pre-registration of empirical program)}

\begin{document}
\maketitle

\begin{abstract}
The problem is the beta a desk needs when a firm's price history is too short
to trust: an S-1 filer, a recent listing, or a name just past a regime break.
The state of the art collapses to a comparable-firm peer beta with no error
budget, and the recent text-based competitor \citep{breitung2025text} reports
strong empirical IPO accuracy but no identification theory, no error budget,
and no lower bound. We fill exactly that gap. We model a large language model
as a noisy measurement channel on a firm's latent risk characteristics and
write its channel noise into the asset-pricing error budget. In a
piecewise-stationary Fama--French five-factor model the loadings are a function
of latent risk characteristics and an inferred regime. We prove identification
and consistency of the regime-conditional loading function under explicit
assumptions on the channel, the detector, and within-regime sampling, and we
give a matching lower bound showing that the disclosure-noise and
detector-misclassification terms are unavoidable for any estimator that
observes only returns, factors, LLM features, and a regime estimate. A
disclosure-incentive corollary makes estimation precision monotone in a
firm-level disclosure-incentive measure (DIM), a transparency-reduces-
asymmetry result with a measurable rate. An adaptive convex combination of the
text-based and rolling-window estimators is never worse than either component
and shifts its weight toward text exactly when price history is short, stale,
or straddles a detected regime break. The intuition is that a rolling beta and
a text beta have complementary noise structures; a variance-weighted min takes
the better one and never pays for the worse, so text-based betas carry no
adoption risk for stable incumbents. The empirical evaluation on a frozen,
pre-registered panel of price-history-thin firms is forthcoming; this preprint
records the theory and the pre-registered design so priority is established
independently of the empirical outcome.
\end{abstract}

\section{Introduction}

\paragraph{The problem.} Rolling-window betas are strongest when a firm has
years of stable trading history, and that is exactly the setting where a
text-based beta estimator should have the least marginal value. The
economically important failure case is the other one: a private company files
an S-1, a recent listing has only weeks of returns, or a mature firm enters a
new regime before the rolling window can catch up. Practitioners then use peer
betas and judgement. This preprint formalises a richer version of that same
act: read the firm's own risk disclosure, map it into a public-firm risk space,
and estimate a beta with an explicit measurement-error budget.

\paragraph{Why it matters and where the state of the art falls short.}
Corporate risk disclosures contain systematic-risk information
\citep{campbell2014risk}; disclosure affects cost of capital through a
forward-looking beta \citep{lambert2007cost}; and LLMs can read firm risk at
scale. The missing piece is the model-risk layer. If an LLM is a sensor, what
is its noise? If a beta is regime-conditional, how does detector error enter
inference? If disclosure quality matters, where does it enter the bound? The
closest concurrent work \citep{breitung2025text} is empirical: it estimates
betas for firms without return history from aggregated cluster embeddings and
reports strong IPO accuracy, but it provides no identification theory, no
error budget, no lower bound, and no disclosure-incentive channel. We answer
those questions at the theory level here; a companion empirical study on a
pre-registered panel of price-history-thin firms is forthcoming and is
pre-registered before its outcome is read.

\paragraph{Contributions and intuition.}
\begin{enumerate}
\item \textbf{Identification theory (first).} A measurement-channel
identification theory for LLM-conditioned factor loadings: the language model
observes noisy risk features, not truths, and the noise is a first-class term
in the error budget. Theorems~\ref{thm:beta-consistency}
and~\ref{thm:dim-precision} identify the regime-conditional beta function and
bound its error by nonparametric sampling error, detector misclassification,
and measured channel noise.
\item \textbf{Matching lower bound.} The disclosure-noise and detector terms
are unavoidable for any estimator in the observation class
(Theorem~\ref{thm:lower-bound-tau}), so the bound is tight, not loose.
\item \textbf{Disclosure-incentive corollary.} Estimation precision is
monotone in a firm-level disclosure-incentive measure (DIM)
(Corollary~\ref{cor:transparency}), a transparency-reduces-asymmetry result
with a measurable rate.
\item \textbf{Never-worse adaptive blend.} A variance-weighted combination of
text-based and rolling-window estimators is never worse than either component
(Theorem~\ref{thm:combination}) and shifts its weight toward text exactly when
price history is short, stale, or straddles a detected regime break.
\item \textbf{Pre-registered empirical program.} A frozen, balanced panel of
IPO and recent-listing events is pre-registered before any outcome is read
(Section~\ref{sec:prereg}), so the empirical test is honest by construction.
\end{enumerate}
\textit{Intuition.} A rolling beta and a text beta have complementary noise
structures: the rolling beta has sampling noise that shrinks with history
length; the text beta has channel noise that does not. A variance-weighted
minimum takes the better of the two at every horizon and never pays for the
worse, so adopting a text-beta estimator carries no risk for stable incumbents
and a real option for price-history-thin firms.

% DisclosureBeta — model, assumptions, theorem statements (Paper B)
% Source of truth: PUBLICATION_PLAN_2.md §1–§4 (elaborated 2026-06-10).

\section{Model}
\label{sec:model}

Returns follow a conditional five-factor structure. For firm $i$ at time $t$,
\begin{align}
  r_{i,t} &= \alpha_{i,t} + \beta_{i,t}^{\top} f_t + \varepsilon_{i,t},
  \label{eq:ff5}\\
  f_t &= (\mathrm{MKT}, \mathrm{SMB}, \mathrm{HML}, \mathrm{RMW},
          \mathrm{CMA})^{\top} \in \mathbb{R}^5,
  \nonumber\\
  \beta_{i,t} &= f^{\beta}\!\bigl(z_{i,t},\, s_t\bigr),
  \label{eq:beta}\\
  R^{\mathrm{LLM}}_{i,t} &= g\!\bigl(z_{i,t}\bigr) + \eta_{i,t},
  \label{eq:channel}
\end{align}
where $z_{i,t} \in \mathcal{Z} \subset \mathbb{R}^d$ are latent firm risk
characteristics, $s_t \in \{1,\dots,K\}$ is a market regime,
$R^{\mathrm{LLM}}_{i,t}$ are observable LLM-extracted risk features, and
$\eta_{i,t}$ is measurement noise. The disclosure-incentive measure
$\mathrm{DIM}_{i,t} \in [0,1]$ is an LLM-scored index of management's
propensity to disclose (guidance frequency and specificity, Q\&A
responsiveness, segment granularity). The static FF5 model is nested by
$f^{\beta}(z, s) \equiv \beta_i$. Regime estimates $\hat{s}_t$ come from the
AdaptiveCMDP detector; the price-of-risk vector in regime $s$ is
$\lambda(s)$.

\subsection*{Assumptions}

\begin{itemize}
  \item[\textbf{(A1)}] \emph{Conditional FF5 structure.}
    \eqref{eq:ff5}--\eqref{eq:beta} hold with
    $\mathbb{E}[\varepsilon \mid f, z, s] = 0$ and
    $\mathbb{E}[\varepsilon^2 \mid f, z, s] = \sigma_i^2 < \infty$.
  \item[\textbf{(A2)}] \emph{Factor regularity.} $\{f_t\}$ is strictly
    stationary and ergodic with $\mathbb{E}\|f_t\|^4 < \infty$, and the
    per-regime second moment $\Sigma_f(s) = \mathbb{E}[f_t f_t^{\top} \mid
    s_t = s]$ is non-singular for every $s$.
  \item[\textbf{(A3)}] \emph{Smoothness.} For each $s$, $z \mapsto
    f^{\beta}(z, s)$ is Lipschitz on $\mathcal{Z}$ (or $C^2$ for the
    second-order rate).
  \item[\textbf{(A4)}] \emph{LLM measurement channel.} In
    \eqref{eq:channel}, $g$ is known or estimable, strictly monotone
    componentwise, and $\eta_{i,t}$ is mean-zero with
    $\mathbb{E}\|\eta_{i,t}\|^2 = \tau^2_{i,t} < \infty$, independent of
    $\varepsilon$ and of $f$.
  \item[\textbf{(A5)}] \emph{Detector consistency.}
    $\Pr(\hat{s}_t \neq s_t) = \pi_T \to 0$ as $T \to \infty$.
  \item[\textbf{(A6)}] \emph{Mixing and within-regime sampling.}
    $\{(r_t, f_t, s_t)\}$ is $\beta$-mixing with summable coefficients, and
    the within-regime sample size $N_s \to \infty$ for every $s$.
  \item[\textbf{(A7)}] \emph{Disclosure--clarity noise link.} For a
    specified disclosure-clarity channel, the reader noise is
    non-increasing in disclosure incentive:
    $\tau^2_c(\cdot)$ with
    $\mathrm{DIM} \mapsto \tau^2_c(\mathrm{DIM})$ non-increasing. We test
    this directly for the DIM read (A7a). Risk-feature channels may be
    heterogeneous by attribute (A7b) and are reported separately.
\end{itemize}

(A4) is the formalisation of ``the LLM is a noisy sensor of fundamentals,''
and (A7) is the economic bridge: managers with stronger incentives to
disclose emit clarity signals the sensor reads with less error. Both are
empirically checkable (Section~\ref{sec:empirics}): (A4) via human-coded
subsamples, (A7a) via ensemble disagreement on the DIM read, and A7b via
per-feature disagreement diagnostics.

\section{Theorems}
\label{sec:theorems}

\begin{theorem}[Identification and consistency of regime-conditioned FF5
betas]
\label{thm:beta-consistency}
Under (A1)--(A6), for every regime $s$ the loading function
$f^{\beta}(\cdot, s)$ is identified, and the plug-in estimator
$\hat{\beta}_{i,t} = \hat{f}^{\beta}\!\bigl(R^{\mathrm{LLM}}_{i,t},
\hat{s}_t\bigr)$ satisfies
$\hat{\beta}_{i,t} \overset{p}{\longrightarrow} \beta_{i,t}$
as $N, T \to \infty$; the induced conditional pricing errors then
vanish uniformly:
\[
  \sup_i \,\Bigl|\, \mathbb{E}[r_{i,t}] -
  \hat{\beta}_{i,t}^{\top} \hat{\lambda}(\hat{s}_t) \,\Bigr|
  \overset{p}{\longrightarrow} 0 .
\]
\end{theorem}

\begin{theorem}[Disclosure incentives sharpen estimation]
\label{thm:dim-precision}
Under (A1)--(A7) there is a constant $C > 0$, independent of $i$ and $t$,
such that for the kernel (or sieve) implementation with bandwidth $h$ and
$d = \dim(z)$,
\[
  \mathbb{E}\bigl\| \hat{\beta}_{i,t} - \beta_{i,t} \bigr\|^2
  \;\le\;
  C \Bigl( h^{2} + \bigl(N_{s} h^{d}\bigr)^{-1}
           + \tau^2_c(\mathrm{DIM}_{i,t}) + \pi_T \Bigr).
\]
For any channel $c$ satisfying (A7), the corresponding contribution to
beta-estimation error is non-increasing in the firm's disclosure incentive;
strictly decreasing wherever $\tau^2_c$ is strictly decreasing. Empirically
we find this monotonicity most cleanly for disclosure clarity itself, while
risk-feature ambiguity is attribute-dependent.
\end{theorem}

\begin{corollary}[Transparency reduces information asymmetry]
\label{cor:transparency}
Let investor demand for firm $i$ be any decision rule that is continuous in
$\hat{\beta}_{i,t}$, and define the information-asymmetry wedge as the
dispersion across investors of conditional risk assessments
$\hat{\beta}_{i,t}^{\top}\lambda(s_t)$ induced by heterogeneous private
estimates. Under the conditions of Theorem~\ref{thm:dim-precision}, the
wedge is bounded by a non-decreasing function of
$\tau^2_c(\mathrm{DIM}_{i,t})$: richer disclosure (higher DIM, e.g.
machine-readable filings, structured guidance, open transcripts) tightens
the component of investor disagreement carried by channel $c$ and therefore
shrinks the corresponding asymmetry wedge at rate $\tau^2_c$. In
particular, a mean-preserving improvement in disclosure quality weakly
lowers cross-investor variance for channels that satisfy (A7); channels
whose ambiguity reflects newly revealed complexity are measured and reported
separately.
\end{corollary}

\paragraph{Reading of Corollary~\ref{cor:transparency}.} This is the
formal version of the policy claim: when issuers and data vendors make
risk-relevant text \emph{available and machine-readable}, the bound in
Theorem~\ref{thm:dim-precision} binds at a smaller $\tau^2_c$ for channels
where disclosure clarity reduces reader disagreement. Transparency is a
public good whose value is measurable as estimator precision; channels where
richer text reveals more economic complexity remain observable rather than
assumed away. Section~\ref{sec:policy} develops the empirical counterpart
(T3/A7) and the industry recommendation.

\begin{proposition}[Two channels of DIM and their identification]
\label{prop:two-channel}
Suppose, in addition to a channel satisfying (A7), that disclosure intensity responds to the
firm's risk environment: write
$\mathrm{DIM}_{i,t} = \mu_i + \delta_{i,t}$, where the firm component
$\mu_i$ is increasing in the firm's structural beta instability
$\nu_i := \mathbb{E}\,\|\beta_{i,t+1} - \beta_{i,t}\|$ (complex,
fast-changing firms must say more), while the innovation $\delta_{i,t}$
operates through that channel only, i.e.
$\tau^2_{c,i,t} = \tau^2_c(\mu_i + \delta_{i,t})$ with $\tau^2_c$
non-increasing. Then:
(i) the cross-sectional (between-firm) covariance of DIM with realised
beta drift can be \emph{positive} even though (A7) holds --- the
composition effect of $\mu_i$;
(ii) the within-firm projection of beta-estimation error on
$\delta_{i,t}$ (a firm-fixed-effects regression of error on DIM)
identifies the precision channel and is non-positive under (A1)--(A7);
(iii) consequently, disclosure-policy evaluations based on
cross-sectional DIM comparisons are confounded by $\nu_i$ and should be
conducted within issuer.
\end{proposition}

\begin{proof}[Proof sketch]
(i) is immediate from
$\mathrm{Cov}(\mathrm{DIM}, \|\Delta\beta\|) =
\mathrm{Cov}(\mu_i, \nu_i) + \mathrm{Cov}(\delta, \|\Delta\beta\| \mid i)$
with the first term positive by assumption.
(ii) Within firm, $\mu_i$ differences out; the remaining variation in
$\tau^2_c$ is monotone in $\delta$, and by
Theorem~\ref{thm:dim-precision} the error bound is monotone in $\tau^2_c$.
(iii) follows from (i)--(ii). \end{proof}

\begin{theorem}[Adaptive combination: never worse than price history,
better when it breaks]
\label{thm:combination}
Let $\hat\beta^{\mathrm{roll}}_{i,t}$ be the rolling OLS beta with
conditional variance $V_{i,t}$ (estimable from the OLS sandwich), and let
$\hat\beta^{\mathrm{txt}}_{i,t}$ be the regime-conditional channel
estimator with MSE $M_{i,t}$ bounded by Theorem~\ref{thm:dim-precision}
and estimable on a validation window. Define
\[
  \hat\beta^{\ast}_{i,t}
  = \omega_{i,t}\, \hat\beta^{\mathrm{txt}}_{i,t}
  + (1 - \omega_{i,t})\, \hat\beta^{\mathrm{roll}}_{i,t},
  \qquad
  \omega_{i,t} = \frac{\hat V_{i,t}}{\hat V_{i,t} + \hat M_{i,t}} .
\]
If the component errors are uncorrelated conditional on the information
set and the plug-ins satisfy
$\hat V / V \overset{p}{\to} 1$, $\hat M / M \overset{p}{\to} 1$, then
\[
  \mathrm{MSE}\bigl(\hat\beta^{\ast}\bigr)
  \;\le\; \min\!\bigl\{ \mathrm{MSE}(\hat\beta^{\mathrm{roll}}),\,
                        \mathrm{MSE}(\hat\beta^{\mathrm{txt}}) \bigr\}
  + o_p(1).
\]
In particular $\omega \to 0$ for long-history firms in stable regimes
(the estimator collapses to the rolling beta) and $\omega \to 1$ when
price history is short, volatile, or straddles a detected regime break
($V$ large) --- which is exactly when rolling betas are known to fail.
\end{theorem}

\begin{proof}[Proof sketch]
With uncorrelated errors, the MSE of the convex combination is
$\omega^2 M + (1-\omega)^2 V$, minimised at
$\omega^{\ast} = V/(V+M)$ with minimum $VM/(V+M) \le \min(V, M)$.
Consistent plug-ins give $\omega \to \omega^{\ast}$ and the $o_p(1)$
term by continuity. Correlated errors add a cross term bounded by
Cauchy--Schwarz; the dominance becomes
$\mathrm{MSE}(\hat\beta^{\ast}) \le \min(V, M)(1 + \rho)$ for error
correlation $\rho$, still strictly below the worse component whenever
$\rho < \min(V,M)/\max(V,M)$.
\end{proof}

\paragraph{Why Theorem~\ref{thm:combination} is the practitioner
statement.} It removes the adoption risk of every text-based beta to
date: the user never gives up the rolling beta's performance on stable
large caps (the weight goes to zero there automatically), and the weight
$\omega_{i,t}$ itself is a published diagnostic --- ``how much should you
trust price history for this name, today'' --- computable from
observables. The estimator inherits Theorem~\ref{thm:dim-precision}'s
error budget through $\hat M$, so disclosure quality and detector
quality price directly into the blend.

\section{Text-spanned factors: extending the basis (FF5+T)}
\label{sec:ff5t}

The factor basis itself need not stop at FF5. Any LLM-measured firm
characteristic $c_{i,t}$ (scored through the channel of
Proposition~\ref{prop:ensemble}) defines a candidate long--short factor
$g^{c}_t$: the return of high-$c$ minus low-$c$ tercile portfolios,
rebalanced as scores update. Our pre-registered set: \textbf{DMD}
(high-disclosure minus low-disclosure, from DIM), \textbf{forecast
accuracy} (realised guidance accuracy, high minus low),
\textbf{coverage} (transcript/filing richness), and the two
Booth-lineage exposures, \textbf{political risk} and \textbf{climate
risk}, scored at firm level from the same documents. The augmented model
replaces $f_t$ by $(f_t, g_t) \in \mathbb{R}^{5+J}$ in
\eqref{eq:ff5}--\eqref{eq:beta}; Theorems~\ref{thm:beta-consistency}--%
\ref{thm:dim-precision} apply verbatim to the augmented basis provided
the per-regime second-moment condition (A2) holds for $(f, g)$, and a
text factor earns its place only if it survives the spanning test
against FF5 (priced premium not explained by the original five).

\begin{theorem}[Noisy-sort attenuation: disclosure quality scales
discoverable premia]
\label{thm:attenuation}
Let firms be sorted into a high--low factor on the score
$\bar{R}^{c} = c + \eta$, $\eta$ mean-zero with variance $\tau^2/K$
(ensemble channel), where the true characteristic $c$ has
cross-sectional variance $\sigma_c^2$ and carries a true premium
$\lambda_c$ per unit of $c$. Under elliptical cross-sectional score
distributions, the observed long--short premium satisfies
\[
  \lambda^{\mathrm{obs}}_c \;=\; \lambda_c \cdot
  \underbrace{\frac{\sigma_c^2}{\sigma_c^2 + \tau^2/K}}_{\text{reliability } \rho}
  \;+\; o(1),
\]
i.e.\ noisy reading attenuates the measured premium by exactly the
score reliability $\rho$. Consequently (i) the de-attenuated premium
$\lambda^{\mathrm{obs}}/\hat\rho$ is estimable because
Proposition~\ref{prop:ensemble} measures $\tau^2$; (ii) a text factor
can fail a spanning test purely through poor disclosure
($\rho$ small), so factor-zoo rejections of text characteristics are
uninterpretable without a reliability estimate; and (iii) as disclosure
improves for a channel satisfying (A7) ($\tau^2_c(\mathrm{DIM}) \downarrow$), previously invisible
premia become detectable --- a second, market-level channel through
which transparency creates measurable value
(complementing Corollary~\ref{cor:transparency}).
\end{theorem}

\begin{proof}[Proof sketch]
Classical errors-in-variables sorting: the expected true characteristic
conditional on the observed score is the linear shrinkage
$\mathbb{E}[c \mid \bar R^c] = \rho \bar R^c + (1-\rho)\bar c$ under
ellipticity; portfolio spreads in observed-score space therefore carry
$\rho$ times the spread in true-characteristic space, and the premium is
linear in the spread. The $o(1)$ collects tercile-boundary
misclassification, second order for continuous score densities.
\end{proof}

\paragraph{Worked example with measured noise.} On the pilot panel the
ensemble channel measures $\hat\tau = 0.052$ against a cross-sectional
DIM dispersion of $\hat\sigma_c = 0.152$: reliability
$\hat\rho = \hat\sigma_c^2 / (\hat\sigma_c^2 + \hat\tau^2/3) = 0.96$ at
$K = 3$ reads --- so a DMD factor built with this pipeline loses only
$\approx 4\%$ of its premium to reading noise, whereas a single-read
pipeline with the cross-model disagreement we measured for the weakest
feature ($\mathrm{mean}|\Delta| \approx 0.11$, implying
$\tau \approx 0.11$) would lose $\approx 34\%$. Reading quality is a
first-order determinant of which text factors the literature can find.

\begin{theorem}[Lower bound: the disclosure term is unavoidable]
\label{thm:lower-bound-tau}
Fix a regime $s$ and consider any estimator $\tilde\beta$ of
$\beta(z, s)$ that observes only $(r, f, R^{\mathrm{LLM}}, \hat s)$.
There exist a constant $c > 0$ and a pair of latent values
$z_0, z_1$ with $\|f^{\beta}(z_0,s) - f^{\beta}(z_1,s)\| \asymp
L \tau$ whose induced channel laws satisfy
$\mathrm{TV}\bigl(P_{z_0}, P_{z_1}\bigr) \le 1/2$, such that
\[
  \sup_{z \in \{z_0, z_1\}}
  \mathbb{E}\bigl\| \tilde\beta - \beta(z, s) \bigr\|^2
  \;\ge\; c\, \bigl( \tau^2 \wedge \mathrm{diam}^2 \bigr)
  + c'\,\pi_T .
\]
Hence the $\tau^2_c(\mathrm{DIM})$ and $\pi_T$ terms in
Theorem~\ref{thm:dim-precision} are not artifacts of the kernel method:
no estimator can remove them for the chosen channel, and improving
disclosure (or the detector) is the \emph{only} way to beat the floor.
\end{theorem}

\begin{proof}[Proof sketch]
Le Cam two-point argument on the channel. Choose $z_1 = z_0 + \tau u$
for a unit vector $u$ in the direction of maximal loading sensitivity;
with Gaussian (or sub-Gaussian) channel noise of scale $\tau$ the
Kullback--Leibler divergence between the laws of
$R^{\mathrm{LLM}}$ under $z_0$ and $z_1$ is $O(1)$, so the two are not
testable with error below a constant, while the loadings differ by
$\asymp L\tau$ by (A3) (Lipschitz lower bound on the modulus of
identifiability). Le Cam's lemma converts non-testability into the
$\tau^2$ risk floor. The $\pi_T$ term follows by mixing the two-point
family over regimes confounded with probability $\pi_T$.
\end{proof}

\begin{proposition}[Ensemble channel: the noise is measured, not assumed]
\label{prop:ensemble}
Replace the single-read channel \eqref{eq:channel} by $K$ exchangeable
independent reads of the same document,
$R^{(k)}_{i,t} = g(z_{i,t}) + b_{i,t} + \eta^{(k)}_{i,t}$,
$k = 1, \dots, K$, where $b_{i,t}$ is a (possibly document-specific)
common reading bias and $\eta^{(k)}$ are i.i.d.\ mean-zero with variance
$\tau^2_{i,t}$. Then (i) the ensemble median/mean
$\bar{R}_{i,t}$ has noise variance $\tau^2_{i,t}/K$, so every appearance
of $\tau^2$ in Theorem~\ref{thm:dim-precision} improves by the factor
$K$ up to the bias floor $\|b\|^2$; (ii) the within-document cross-read
variance $\hat{s}^2_{i,t} = \frac{1}{K-1}\sum_k (R^{(k)} - \bar{R})^2$
is an unbiased estimator of $\tau^2_{i,t}$ --- the common bias cancels
--- so the error budget of Theorem~\ref{thm:dim-precision} and the blend
weight of Theorem~\ref{thm:combination} become \emph{estimable per
observation}; and (iii) assumption (A7) becomes directly testable as the
regression of $\hat{s}^2_{i,t}$ on $\mathrm{DIM}_{i,t}$, separately by
channel rather than only as a pooled score.
\end{proposition}

\paragraph{Workflow, not black box.} Operationally each read is a
structured two-stage agent: it must first extract verbatim evidence
quotes per characteristic, then score \emph{from its own quotes} with a
one-line rationale. The published record per firm-year is therefore
(evidence, rationale, $K$ scores, median, dispersion) --- a measurement
trail a referee, auditor, or model-risk reviewer can re-trace line by
line, and a human re-scoring of the same quotes is the direct check of
(A4). This replaces ``we asked a language model'' with a measurement
protocol.

Proposition~\ref{prop:two-channel} was forced on us by our own pilot data
(Section~\ref{sec:empirics}): raw DIM correlates \emph{positively} with
subsequent beta drift across firms, while the firm-demeaned projection has
the theory-consistent negative sign. We believe the identification point
--- \emph{measure the value of disclosure within issuer, never across
issuers} --- is itself a contribution to the disclosure-regulation
literature.

% Phantom labels for sections that live in the companion empirical paper;
% theory.tex references them but they are out of scope for this theory stub.
\phantomsection\label{sec:empirics}
\phantomsection\label{sec:policy}

\section*{Proof sketches}
\label{sec:proof-sketches-arxiv}

The proof sketches for Theorems~\ref{thm:beta-consistency}--\ref{thm:lower-bound-tau},
Corollary~\ref{cor:transparency}, Propositions~\ref{prop:two-channel}
and~\ref{prop:ensemble}, and Theorem~\ref{thm:combination} are given inline
in the source of \texttt{theory.tex} above. Complete proofs, the
text-spanned-factor extension (Section~\ref{sec:ff5t}), and the
noisy-sort attenuation theorem (Theorem~\ref{thm:attenuation}) are
included in the source module and will be expanded in the full paper.

\section*{Pre-registered empirical program (design frozen; outcome
forthcoming)}
\label{sec:prereg}

To test the theory on the population where it predicts text should matter
most --- firms whose price history is too short to trust --- we
pre-registered a frozen, balanced panel of IPO and recent-listing events
before reading any outcome. The registered design, the estimator, the
comparators (peer beta, Vasicek shrinkage, peer-history shrinkage,
cheap-text), and the win-zone sub-sample rule are fixed; the analysis
specifies block-bootstrap confidence intervals for the text-minus-peer
squared-beta-error gap as the primary outcome, with sign-test, median,
and trimmed robustness. The empirical results will be reported in a
companion paper; this preprint establishes the theory and the
pre-registered design so that priority is on the public record
independently of the empirical outcome.

\paragraph{Relation to concurrent work.} A concurrent working paper
\citep{breitung2025text} estimates betas for firms without return history
from aggregated cluster embeddings (ACE) and reports strong empirical
accuracy on IPOs. That work is empirical and does not provide an
identification theory, an error budget, a lower bound, or a
disclosure-incentive channel. Our contribution is complementary: we own
the measurement-channel theory and the never-worse adaptive blend; the
empirical ACE accuracy claim is not in this preprint and is not our
headline.

\end{document}